\documentclass[a4paper,10pt]{article}
\usepackage{amsmath}
\usepackage{amssymb}
\usepackage{amsthm}
\usepackage{pgf,tikz}
\usepackage{titlesec}
\usepackage{mathdots}

\titleformat{\section}{\large\bfseries}{\thesection}{1em}{}
\newtheorem{thm}{Theorem}[section]
\newtheorem{prop}[thm]{Proposition}
\newtheorem{definition}[thm]{Definition}
\theoremstyle{definition}
\newtheorem{remark}[thm]{Remark}

\def\C{{\mathbb C}}
\def\P{{\mathbb P}}

\begin{document}

\begin{center}
{\Large {Quadrirational Yang-Baxter maps and the elliptic Cremona system}}

\vskip10mm
{\large James Atkinson\footnote{
School of Mathematics and Statistics,
University of Sydney, NSW 2006, Australia.
james.l.atkinson@gmail.com} 
and 
Yasuhiko Yamada\footnote{
Department of Mathematics, Faculty of Science,
Kobe University, 657-8501, Japan. yamaday@math.kobe-u.ac.jp
}}
\end{center}

\vskip10mm

\noindent
{\bf Abstract.}
This paper connects the quadrirational Yang-Baxter maps, which are two-dimensional integrable discrete systems of KdV type, and the elliptic Cremona system, which is a higher analogue of discrete Painlev\'e equations associated with $\tilde{E}_8$ symmetry.
This is a natural connection between integrable systems in different dimensions that is outside of the usual paradigm of reductions.
Our approach is based on formulation of both systems in terms of birational Coxeter groups.

\section{Introduction}\label{INTRO}
The birational group of Coble \cite{COBI,COBII} plays a key role in the geometric theory of the discrete Painlev\'e equations \cite{sc,10E9,kmnoy}.
This paper connects Coble's group with a birational group obtained as a generalisation of certain integrable discrete systems of KdV type called quadrirational Yang-Baxter maps and multi-quadratic quad-equations \cite{ABSf,AtkNie,ib}.
The associated abstract groups are different, however, they correspond to Coxeter graphs of variable extent (depending on some indices), which allows the following characterisation:
Coble's group is obtained as a parabolic subgroup of the group associated with the quadrirational maps, which in turn is obtained as the centralizer of a certain parabolic subgroup in Coble's group.

The connection between birational groups is the technical result given in this paper, its significance is because they have different points of origin.
The immediate benefits are, (i) a rational realisation of Coble's group in certain geometric (un-normalised) variables that were known previously only through parameterisation in compatible linear actions, and, (ii) the natural (seed, or background) elliptic solutions of the group associated with the quadrirational maps.

Broader consideration of Coble's group and its relation to the Painlev\'e equations has led to the elliptic Cremona system in \cite{kmnoy}, which is recalled in Section \ref{Painleve}.
The group associated with the quadrirational maps is recalled in Section \ref{quadrirational}.
The connection between the groups on the level of birational actions is established in Section \ref{twogroups}. 
In Section \ref{solution} the contact is made on the level of the compatible linear actions (solutions) that enter through parameterisation in terms of the elliptic functions.

\section{Coble's group and the elliptic Painlev\'e equation}\label{Painleve}

The elliptic difference Painlev\'e equation, the master equation among the second order discrete Painlev\'e equations, was discovered and constructed by Sakai \cite{sc}. The fundamental device in the construction is the birational representation of the affine Weyl group $\tilde{E}_8$. The elliptic Painlev\'e equation is  obtained as the translation part of the affine Weyl group.
Following \cite{kmnoy}, we recall the general setting of Coble's birational actions on point-sets in projective space, which contains the affine Weyl group $\tilde{E}_8$ action as a special case.

Let $X\in {\rm Mat}_{j+1,i+j+2}(\C)$
be a configuration of $i+j+2$ points (columns of $X$) on $\P^j$ and the natural coset space $\mathcal{M}_{i,j}$ given by
\begin{equation}
\mathcal{M}_{i,j}={\rm PGL}(j+1,\C) \setminus {\rm Mat}_{j+1,i+j+2}(\C) /{(\C^{*})^{i+j+2}}.
\end{equation}
We have an open chart of $\mathcal{M}_{i,j}$ whose coordinates are given by the canonical form
\begin{equation}\label{Xcan}
X_{\rm can}=\left[\begin{array}{ccccccccc}
0&\cdots&0&0&1&1&1&\cdots&1\\
0&\cdots&0&1&0&1&x_{11}&\cdots&x_{i1}\\
0&\cdots&1&0&0&1&x_{12}&\cdots&x_{i2}\\
\vdots&\iddots&\vdots&\vdots&\vdots&\vdots&\vdots&\vdots&\vdots\\
1&\cdots&0&0&0&1&x_{1j}&\cdots&x_{ij}
\end{array}\right].
\end{equation}

On the space $\mathcal{M}_{i,j}$, there exist natural 
birational actions $w_0, w_1, \ldots, w_{i+j+1}$, where
$w_0$ is the inversion $x_{mn} \to 1/{x_{mn}}$ and $w_n$ ($n\in \{1, \ldots, i+j+1\}$)  arises as the permutation of $n$-th and
$(n+1)$-th column of the matrix $X_{\rm can}$. 
The explicit actions are easily computed on the chart (\ref{Xcan}) and give the following
\begin{definition}\label{coble}
Coble's actions \cite{COBI,COBII} on entries of $X_{\rm can}$ (\ref{Xcan})
corresponding to Figure \ref{cdd3}: 
\begin{equation}\label{coble-actions}
\begin{array}{rll}
w_0:&x_{mn} \rightarrow 1/x_{mn}, & m\in\{1,\ldots,i\}, \ n\in\{1,\ldots,j\}, \\
w_{j+1}:&x_{mn} \rightarrow 1-x_{mn}, & m\in\{1,\ldots,i\}, \ n\in\{1,\ldots,j\}, \\
w_j:&x_{m1} \rightarrow 1/x_{m1}, x_{mn} \rightarrow x_{mn}/x_{m1}, & m\in\{1,\ldots,i\}, \ n\in\{2,\ldots,j\},\\
w_{j+1-n}:&x_{m(n-1)} \leftrightarrow x_{mn}, & m\in\{1,\ldots,i\}, \ n\in\{2,\ldots,j\},\\
w_{j+2}:&x_{1n} \rightarrow 1/x_{1n}, x_{mn} \rightarrow x_{mn}/x_{1n}, & m\in\{2,\ldots,i\}, \ n\in\{1,\ldots,j\},\\
w_{j+1+m}:&x_{(m-1)n} \leftrightarrow x_{mn}, & m\in\{2,\ldots,i\}, \ n\in\{1,\ldots,j\}.\\
\end{array}
\end{equation}
\end{definition}
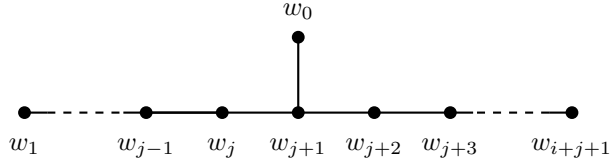
\begin{figure}[t]
\begin{center}
\begin{tikzpicture}[thick]
  \tikzstyle{every node}=[draw=none,minimum size=0pt,inner sep=0pt];
  \node (s1) at (5,1) [label={[label distance=9pt]below:$w_j$}]{};
  \node (s2) at (4,1) [label={[label distance=9pt]below:$w_{j-1}$}]{};
  \node (sj) at (2.4,1) [label={[label distance=9pt]below:$w_1$}]{};
  \node (r0) at (6,1) [label={[label distance=9pt]below:$w_{j+1}$}]{};
  \node (r1) at (6,2) [label={[label distance=6pt]above:$w_0$}]{};
  \node (t1) at (7,1) [label={[label distance=9pt]below:$w_{j+2}$}]{};
  \node (t2) at (8,1) [label={[label distance=9pt]below:$w_{j+3}$}]{};
  \node (ti) at (9.6,1) [label={[label distance=9pt]below:$w_{i+j+1}$}]{};
  \tikzstyle{every node}=[draw,circle,fill=black,minimum size=4pt,inner sep=0pt];
  \draw (s1) node{};
  \draw (s2) node{};
  \draw (sj) node{};
  \draw (r0) node{};
  \draw (r1) node{};
  \draw (t1) node{};
  \draw (t2) node{};
  \draw (ti) node{};
  \draw (s2)--(s1)--(r0)--(t1)--(t2);
  \draw (r0)--(r1);
  \draw (sj)[solid]--(2.7,1);
  \draw (2.7,1)[dashed]--(3.7,1);
  \draw (3.7,1)[solid]--(s1);
  \draw (ti)[solid]--(9.3,1);
  \draw (9.3,1)[dashed]--(8.3,1);
  \draw (8.3,1)[solid]--(t2);
\end{tikzpicture}
\end{center}
\caption{Coxeter-graph for the Coble actions.}
\label{cdd3}
\end{figure}
Linear actions compatible with the birational ones are via a substitution through the function $u \mapsto [u]$, characterised as an odd function satisfying the Riemann relation (see proof below), which in the generic case corresponds to Weierstrass' $\sigma$ function.
\begin{prop}[\cite{kmnoy}]\label{ee}
On the variables
\begin{equation}\label{epsilon}
\epsilon_0,\ldots,\epsilon_{i+j+2},
\end{equation}
the linear actions
\begin{equation}\label{w-linear}
\begin{array}{rll}
w_0:& \epsilon_n \rightarrow \alpha_0+\epsilon_n, \quad& n\in\{0,\ldots,j+1\},\\
w_n:& \epsilon_n \leftrightarrow \epsilon_{n+1}, \quad& n\in\{1,\ldots,i+j+1\},
\end{array}
\end{equation}
where $\alpha_0:=(j-1)\epsilon_0-\epsilon_1 - \epsilon_2 \cdots - \epsilon_{j+1}$,
are compatible with Coble's actions (Definition \ref{coble}), via the substitution
\begin{equation}\label{esub}
x_{mn}=\frac
{[\alpha_0+\epsilon_{j+1-n,j+m+2}][\alpha_0+\epsilon_{j+1,j+2}][\epsilon_{j+1,j+m+2}][\epsilon_{j+1-n,j+2}]}
{[\alpha_0+\epsilon_{j+1,j+m+2}][\alpha_0+\epsilon_{j+1-n,j+2}][\epsilon_{j+1-n,j+m+2}][\epsilon_{j+1,j+2}]},
\end{equation}
where $\epsilon_{a,b}=\epsilon_a-\epsilon_b$, $m\in \{1,\ldots,i\}$ and $n\in\{1,\ldots,j\}$.
\end{prop}
\begin{proof}
The compatibilities are verified without any condition on the function $u\mapsto[u]$ for all actions except: 
$w_0,w_j,w_{j+2}$ that are satisfied if it is assumed $[-u]=-[u]$, and $w_{j+1}$ where the same assumption reduces compatibility to the Riemann relation,
\begin{equation}
\begin{split}
[u+v][u-v][w+x][w-x]
&+[u+w][u-w][x+v][x-v]\\
&\qquad =[u+x][u-x][w+v][w-v].
\end{split}
\end{equation}
Note that the non-trivial actions on $\alpha_0$ that are a consequence of (\ref{w-linear}) are $w_0:\alpha_0\rightarrow -\alpha_0$ and $w_{j+1}:\alpha_0\rightarrow \alpha_0+\epsilon_{j+1,j+2}$.
\end{proof}
\begin{remark} \label{normal}
Though the space $\mathcal{M}_{i,j}$ is a natural compactification of the open space $\C^{ij}$ of the data (\ref{Xdef}),
there is another compactification given by
\begin{equation}
\mathcal{M}'_{i,j}={\rm PGL}(2,\C)^{\otimes j} \setminus {\rm Mat}_{j,i+3},
\end{equation}
where ${\rm PGL}(2,\C)^{\otimes j}$ acts on the elements of $  {\rm Mat}_{j,i+3}$ as M\"obius transformations.
$\mathcal{M}'_{i,j}$ is a configuration space of $i+3$ points in $(\P^1)^j$ and the array (\ref{Xdef}) can be viewed as
the coordinate of the following element
\begin{equation}
X'_{\rm can}=\left[\begin{array}{cccccc}
\infty&0&1&x_{11}&\cdots&x_{i1}\\
\vdots&\vdots&\vdots&\vdots&\ddots&\vdots\\
\infty&0&1&x_{1j}&\cdots&x_{ij}\\
\end{array}
\right] \in \mathcal{M}'_{i,j}.
\end{equation}
On the space $\mathcal{M}_{i,j}$, the action $w_m$ ($m \neq j$) has the following simple meaning.
$w_m$ ($m \in \{1,\ldots,j-1\}$): exchange of $m$-th and $(m+1)$-th row,
$w_0$: exchange of 1st and 2nd columns, and $w_{j+n}$ ($n \in \{1,\cdots,i+1\}$): exchange of $n$-th and $(n+1)$-th columns.
This picture $\mathcal{M}'_{i,j}$ will be useful for geometric interpretation of Propositions \ref{norm} and \ref{tail}.
\end{remark}

\begin{remark}
The actions (\ref{coble-actions}) give the affine Weyl group of type $\tilde{E}_8$ for $(i,j)=(5,2)$.
In the theory of Painlev\'e equations, however, we consider the case $(i,j)=(6,2)$, i.e. configuration of 10 points in $\P^2$. 
In the Painlev\'e context, since the 9 points $p_1, \ldots, p_9$ and the last one $p_{10}$ play different roles (parameters and unknown variables respectively), we omit the last generator $w_9$, then we have the affine Weyl group of type $\tilde E_8$ \cite{WE10}.
Note that the generic 9 points $p_1, \ldots, p_9$ can be parametrized as (\ref{esub}), but not for the last one. 
Since the last point $p_{10}$ is generically not on the cubic curve determined by the 9 points $p_1, \ldots, p_9$.
\end{remark}

\section{Birational group associated with quadrirational maps}\label{quadrirational}
This section defines and gives the basic features of a birational group that emerged as the generalisation of Yang-Baxter maps \cite{ABSf} and multi-quadratic quad-equtions \cite{AtkNie} given in \cite{ib}.
\begin{definition}\label{actions}
For positive integers $i$ and $j$, introduce actions on variables in the array
\begin{equation}
\left[\begin{array}{cccc}
x_{11} & x_{21} & \cdots & x_{i1}\\
x_{12} & x_{22} & \cdots & x_{i2}\\
\vdots & \vdots & & \vdots \\
x_{1j} & x_{2j} & \cdots & x_{ij}
\end{array}\right],\label{Xdef}
\end{equation}
as follows: 
\begin{equation}\label{rational_actions}
\begin{array}{rll}
a_m: & x_{(m-1)n} \leftrightarrow x_{mn}, & m \in \{2,\ldots,i\},\  n\in\{1,\ldots,j\},\\
b_n: & x_{m(n-1)} \leftrightarrow x_{mn}, & m \in \{1,\ldots,i\},\  n\in\{2,\ldots,j\},\\
\sigma: & x_{mn}\rightarrow q(x_{11},x_{mn},x_{m1},x_{1n}), \ & m\in\{2,\ldots,i\},\  n\in\{2,\ldots,j\},
\end{array}
\end{equation}
where trivial actions are omitted, and $q$ is a given rational expression (see below).
Also introduce two derived actions,
\begin{equation}\label{more_actions}
a_1:=\sigma b_2 \sigma , \quad b_1:=\sigma a_2 \sigma.
\end{equation}
\end{definition}

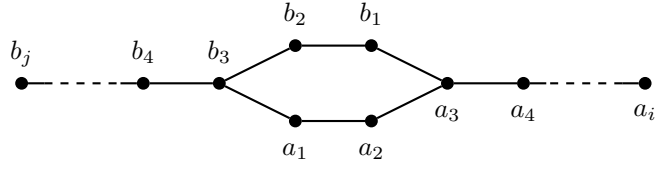
\begin{figure}[t]
\begin{center}
\begin{tikzpicture}[thick]
  \tikzstyle{every node}=[draw,circle,fill=black,minimum size=4pt,inner sep=0pt];
  \node (tt1) at (5,0.5) [label={[label distance=3pt]below:$a_1$}]{};
  \node (t1) at (5,1.5) [label={[label distance=3pt]above:$b_2$}]{};
  \node (t2) at (4,1) [label={[label distance=3pt]above:$b_3$}]{};
  \node (t3) at (3,1) [label={[label distance=3pt]above:$b_4$}]{};
  \node (t4) at (1.4,1) [label={[label distance=2pt]above:$b_j$}]{};
  \node (ss1) at (6,1.5) [label={[label distance=3pt]above:$b_1$}]{};
  \node (s1) at (6,0.5) [label={[label distance=3pt]below:$a_2$}]{};
  \node (s2) at (7,1) [label={[label distance=3pt]below:$a_3$}]{};
  \node (s3) at (8,1) [label={[label distance=3pt]below:$a_4$}]{};
  \node (s4) at (9.6,1) [label={[label distance=3pt]below:$a_i$}]{};
  \draw (t2)--(tt1)--(s1)--(s2)--(s3);
  \draw (s2)--(ss1)--(t1)--(t2)--(t3);
  \draw (t4)[solid]--(1.7,1);
  \draw (1.7,1)[dashed]--(2.7,1);
  \draw (2.7,1)[solid]--(t3);
  \draw (s4)[solid]--(9.3,1);
  \draw (9.3,1)[dashed]--(8.3,1);
  \draw (8.3,1)[solid]--(s3);
\end{tikzpicture}
\end{center}
\caption{
Coxeter-graph associated with the quadrirational maps. 
}
\label{cdd2}
\end{figure}

The actions of Definition \ref{actions} are to be associated with the relations encoded in the Coxeter graph of Figure \ref{cdd2}, 
where relations involving $\sigma$ correspond to the diagram automorphism extending (\ref{more_actions}): $\sigma^2=(\sigma a_{m})^2=(\sigma b_{n})^2=1$ for $m,n>2$.
The encoded relations are satisfied, when the rational expression $q$ appearing in (\ref{rational_actions}), is one of the following:
\begin{eqnarray}
{\textrm F_{I}:} && q(w,x,y,z)=
yz\frac{(w-1)(x-1)-(y-1)(z-1)}{yz(w-1)(x-1)-wx(y-1)(z-1)},\label{F1q}\\
{\textrm F_{II}:} && q(w,x,y,z)=yz\frac{w+x-y-z}{wx-yz},\label{F2q}\\
{\textrm F_{IV}:} && q(w,x,y,z)=y+z-\frac{wx-yz}{w+x-y-z}.\label{F4q}
\end{eqnarray}
These expressions are related, respectively, to the systems,
\begin{eqnarray}
&x_1x_2x_3=y_1y_2y_3,\ (1-x_1)(1-x_2)(1-x_3) = (1-y_1)(1-y_2)(1-y_3), \label{F1a}\\
&x_1+x_2+x_3=y_1+y_2+y_3,\ x_1x_2x_3 = y_1y_2y_3, \label{F2a}\\
&x_1+x_2+x_3=y_1+y_2+y_3,\ x_1^2+x_2^2+x_3^2 = y_1^2+y_2^2+y_3^2. \label{F4a}
\end{eqnarray} 
Specifically, each of the three systems can be written differently in terms of the corresponding expression $q$, as
\begin{equation}\label{map}
x_3=q(x_1,x_2,y_1,y_2), \ y_3=q(y_1,y_2,x_1,x_2).
\end{equation}
In turn, (\ref{F1a}), (\ref{F2a}) and (\ref{F4a}) can be characterised more invariantly, as the condition for the linear denendence of the three polynomials
\begin{equation}\label{polys}
(x-x_1)(x-x_2)(x-x_3), \ 
(x-y_1)(x-y_2)(x-y_3), \ 
P(x),
\end{equation}
where, respectively, $P(x)=x(1-x)$, $P(x)=x$ and $P(x)=1$.
For $q$ to be uniquely obtained and the subsequent group relations to hold, $P(x)$ can be any non-zero polynomial of degree three or less.
M\"obius changes of variables allow to fix coefficients of $P(x)$ without losing generality, but will not change its number of roots, so three canonical forms are required.

In the simplest non-trivial case of the group, when $i=j=2$, the relations are not obvious from the definition, but encode the invariance of system (\ref{map}) under permutations of the variables. 
The permutation symmetry is clear from the expressions (\ref{F1a}), (\ref{F2a}) and (\ref{F4a}), or from the characterisation in terms of linearly dependent polynomials (\ref{polys}).
For instance, that $\sigma^2=1$ follows from invariance of (\ref{map}) under the permutation $x_2\leftrightarrow x_3$.
In general, the group relations are satisfied for any $i,j$ as a consequence of the instance $i=2$, $j=3$, which can therefore be taken as the underlying consistency property.
It can be encoded in a mnemonic associated with a 5-simplex by associating variables with edges, system (\ref{map}) to pairs of opposing faces, and initial data to all edges along a Hamiltonian cycle \cite{ib}.

The terminology of $F_I$, $F_{II}$ and $F_{IV}$ in (\ref{F1q}), (\ref{F2q}) and (\ref{F4q}) is a result of the origin of these systems in the Alder-Bobenko-Suris classification of quadrirational maps \cite{ABSf} which is a different characterisation to the one above.
A system of two polynomial equations in four variables labelled by edges of a quad, is called quadrirational, if it determines the variables adjacent to any vertex rationally from the remaining two. 
System (\ref{map}) satisfies this definition for any fixed choice of $x_3$ and $y_3$, by associating $x_1$ with a quad edge opposite to $x_2$ and $y_1$ with an edge opposite to $y_2$.
The integrability is related to consistent extension of this construction from the quad, to a hypercube. 
Though not originally formulated as a birational group, it corresponds to Definition \ref{actions} with $i=2$ restricted to the subset of generators $\sigma,b_2,\ldots,b_j$.
The underlying consistency property for the quadrirational maps, associated with the cube, has the geometric interpretation as an incidence theorem related to pencils of conics in $\mathbb{P}^2$.

\begin{remark}
Up to natural transformations, there are two remaining quadrirational maps $F_{III}$ and $F_{V}$ found in \cite{ABSf}, which can be obtained locally as limiting cases from the ones listed above.
Those systems can be formulated as a birational Coxeter group associated with symmetries of the hypercube, however, the limiting procedure breaks the symmetry between variables and parameters of the map, and is incompatible with extension to the more general Coxeter graph.
\end{remark}

\section{Relation between the groups}\label{twogroups}
The birational group associated with the quadrirational maps that has been described in the previous section, is connected with Coble's group in two ways. 
\begin{prop}\label{norm}
Consider Definition \ref{actions} with $i>3$, $q$ as in (\ref{F1q}), (\ref{F2q}) or (\ref{F4q}), and introduce new variables
\begin{equation}\label{zdef}
z_{mn} = \frac{(x_{(m+3)n}-x_{2n})(x_{1n}-x_{3n})}{(x_{(m+3)n}-x_{1n})(x_{2n}-x_{3n})},
\end{equation}
where $m\in\{1,\ldots,i-3\}$ and $n\in\{1,\ldots,j\}$.
The actions $a_2,\ldots,a_{i},b_1,\ldots,b_{j}$ on $x_{11},\ldots,x_{ij}$, induce the following actions on the new variables:
\begin{equation} \label{induced}
\begin{array}{rll}
a_2:&z_{mn} \rightarrow 1/z_{mn}, & m\in\{1,\ldots,i-3\}, \ n\in\{1,\ldots,j\}, \\
a_3:&z_{mn} \rightarrow 1-z_{mn}, & m\in\{1,\ldots,i-3\}, \ n\in\{1,\ldots,j\}, \\
b_1:&z_{m1} \rightarrow 1/z_{m1}, z_{mn} \rightarrow z_{mn}/z_{m1}, & m\in\{1,\ldots,i-3\}, \ n\in\{2,\ldots,j\},\\
b_n:&z_{m(n-1)} \leftrightarrow z_{mn}, & m\in\{1,\ldots,i-3\}, \ n\in\{2,\ldots,j\},\\
a_4:&z_{1n} \rightarrow 1/z_{1n}, z_{mn} \rightarrow z_{mn}/z_{1n}, & m\in\{2,\ldots,i-3\}, \ n\in\{1,\ldots,j\},\\
a_{m+3}:&z_{(m-1)n} \leftrightarrow z_{mn}, & m\in\{2,\ldots,i-3\}, \ n\in\{1,\ldots,j\}.\\
\end{array}
\end{equation}
\end{prop}
These are the Coble actions in Definition \ref{coble} with the re-labelling $x_{mn}\rightarrow z_{mn}$, $i\rightarrow i-3$, $w_0\rightarrow a_2$, $w_{j+1}\rightarrow a_3$, $w_j\rightarrow b_1$, $w_{j+1-n}\rightarrow b_n$ ($n>1$), $w_{j+2}\rightarrow a_4$, $w_{j+1+m}\rightarrow a_{m+3}$ ($m>1$).
There is no action induced by $\sigma$ or $a_1$, so Coble's group is realised as a subgroup of the original one.
Notice also the induced actions are independent of the choice made for $q$.

The substitution (\ref{zdef}) corresponds to the M\"obius change of variables on each row of the array (\ref{Xdef}) that sends the first three entries to $\infty$, $0$ and $1$, respectively, which is the projective normalisation admissible when considering (\ref{Xdef}) as a set of $i$ points in $(\mathbb{P}^1)^{j}$, in the way described in Section \ref{Painleve} (Remark \ref{normal}).
The actions induced by the row and column permutations $a_2,\ldots,a_i,b_2,\ldots,b_j$, are verified by inspection.
In the Painlev\'e theory the action of $b_1$ on $x_{mn}$ has been introduced on $z_{mn}$ directly, or on $x_{mn}$ through the associated linear actions that were described in Section \ref{Painleve}; the precise connection to the above Proposition is given in Section \ref{solution}.
To confirm the induced action of $b_1$, it can be assumed that $i=5$ and $j=2$ without losing generality, so that verification is through direct calculation.

The second connection between the groups is a decomposition of actions from Definition \ref{actions} in terms of Coble's actions (\ref{coble-actions}).
It applies only in the primary case $F_I$.
\begin{prop}\label{decomp}
In the case of $F_I$ (\ref{F1q}), the actions of Definition \ref{actions} decompose in terms of the Coble actions (Definition \ref{coble}):
\begin{equation}\label{abstr}
\begin{array}{ll}
a_m=w_{j+1+m}, \quad & m \in \{2,\ldots,i\},\\
b_n=w_{j+1-n}, \quad & n \in \{2,\ldots,j\},\\
\sigma =(w_0 w_{j+2} w_{j+1})^2w_jw_{j+1}w_0w_{j+2}w_{j+1}w_j.
\end{array}
\end{equation}
\end{prop}
\begin{proof}
The decomposition for $\sigma$ can be verified by calculation in the particular case $i=j=2$ without losing generality.
Note an alternative expression $\sigma =(w_0 w_{j} w_{j+1})^2w_{j+2}w_{j+1}w_0w_{j}w_{j+1}w_{j+2}$ obtained from (\ref{abstr}) by interchanging $w_j$ and $w_{j+2}$.
It coincides with the original expression due to the group relations encoded in Figure \ref{cdd3}.
\end{proof}
\begin{remark}
The relations established in Propositions \ref{norm} and \ref{decomp} are between birational realisations of Coxeter groups, but now restrict attention to only the abstract groups.
Proposition \ref{norm} simply corresponds to realisation of the Coxeter group associated with Figure \ref{cdd3} as a parabolic subgroup of the one associated with Figure \ref{cdd2}.
On the other hand, Proposition \ref{decomp} corresponds to a procedure for obtaining one Coxeter group from another that relies on a decomposition that exists for the {\it centralizer} of a parabolic subgroup $W$ \cite{Nui} (see also \cite{Bri,BH,Allcock}).
This decomposition is the product of three factors: the centre of $W$, a Coxeter group whose associated graph can be obtained algorithmically by modifying the original graph, and a group that corresponds to automorphisms for the modified graph.
In the abstract group $\langle w_0,\ldots,w_{i+j+1}\rangle$ defined by its presentation in terms of relations encoded in Figure \ref{cdd3}, the centralizer of $\langle w_0,w_{j+1}\rangle$ is the semidirect product $\langle a_1,\ldots,a_i,b_1,\ldots,b_j\rangle\rtimes\langle\sigma\rangle$, in terms of elements (\ref{abstr}) and $a_1$, $b_1$ given in terms of elements (\ref{abstr}) by (\ref{more_actions}).
There is no first factor here, because $\langle w_0,w_{j+1}\rangle$ has trivial centre.

It follows from the characterisation in Section \ref{quadrirational} that any M\"obius change of variables that permutes the roots (that share multiplicity) of polynomial $P(x)$ in (\ref{polys}), commutes with the corresponding action $\sigma$, and therefore the whole group of Definition \ref{actions}.
For polynomial $P(x)=x(1-x)$ ($F_{I}$) the group of such transformations is exactly $\langle w_0,w_{j+1}\rangle$ from the Coble actions.
(For $P(x)=x$ ($F_{II}$) it corresponds to the M\"obius subgroup $\{x\mapsto \alpha x: \alpha\neq 0\}$ and for $P(x)=1$ (F$_{IV}$) the M\"obius subgroup that fixes $\infty$, $\{x\mapsto \alpha x + \beta: \alpha\neq 0 \}$).
\end{remark}

\section{Solutions}\label{solution}
In this section linear actions compatible with the birational group associated with the quadrirational maps (Section \ref{quadrirational}) are given.
The corollary from results in Sections \ref{Painleve} and \ref{twogroups} is the solution when $q$ in (\ref{rational_actions}) is $F_{I}$ (\ref{F1q}).
The remaining cases have been obtained by a limiting procedure, and a direct proof is therefore given.
\subsection{Preliminaries}
As in Section \ref{Painleve}, solutions are in terms of the function $[u]$ which is either the Weierstrass sigma function
$[u]=\sigma(u)$ (elliptic case), $[u]={\rm sinh}(u)$ (trigonometric case) or  $[u]=u$ (rational case).
In this section use will also be made of associated functions defined in terms of $[u]$ by the equations
\begin{equation}
\zeta(u) = \frac{[u]'}{[u]}, \quad \wp(u) = \frac{([u]')^2-[u][u]''}{[u]^2}.
\end{equation}
The notation used for these functions corresponds to the elliptic case, but trigonometric and rational cases are included by replacements as follows:
\begin{equation}
\begin{array}{cclcl}
\hline
{\rm ell}&\rightarrow&{\rm trig}&\rightarrow&{\rm rat}\\
\hline
\sigma(u) &\rightarrow& {\mathrm{sinh}}(u) &\rightarrow& u \\
\zeta(u) &\rightarrow& {\mathrm{cosh}}(u)/{\mathrm{sinh}}(u) &\rightarrow& 1/u\\
\wp(u) &\rightarrow& 1/{\mathrm{sinh^2}}(u) &\rightarrow& 1/u^2\\
\hline
\end{array}
\end{equation}
\begin{remark}
The limiting procedure in the chain $F_{I} \rightarrow F_{II} \rightarrow F_{IV}$ (see \cite{ABSf}, or one can think of coalesence of roots of $P(x)$ in (\ref{polys})) does not correspond to the limiting procedure elliptic $\rightarrow$ trigonometric $\rightarrow$ rational. 
Rational, trigonometric and elliptic solutions are given for the birational group corresponding to each case $F_{I}$, $F_{II}$ and $F_{IV}$.
\end{remark}

\subsection{Linear actions compatible with $F_{I}$, $F_{II}$ and $F_{IV}$}
The result corresponding to Proposition \ref{ee}, but for the birational group defined in Section \ref{quadrirational}, is as follows.
\begin{prop}\label{q-solution}
On the variables
\begin{equation}\label{khvars}
k_1,\ldots,k_i,\  h_1,\ldots,h_j,
\end{equation}
the linear actions
\begin{equation}\label{la}
\begin{array}{rll}
a_{m}:& k_{m-1}\leftrightarrow k_m, & m\in\{2,\ldots,i\},\\
b_{n}:& h_{n-1}\leftrightarrow h_n, & n\in\{2,\ldots,j\},\\
\sigma: & k_{m}\rightarrow l_1+l_2-h_1-k_m, & m\in\{2,\ldots,i\},\\
        & h_{n}\rightarrow l_1+l_2-k_1-h_n, & n\in\{2,\ldots,j\},
\end{array}
\end{equation}
are compatible with the corresponding birational actions of Definition \ref{actions}, in the case of F$_{I}$ (\ref{F1q}) via the substitution
\begin{equation}\label{F1subs}
x_{mn}=\frac{[k_m+h_n-l_1][l_1][k_m-l_2][h_n-l_2]}{[k_m+h_n-l_2][l_2][k_m-l_1][h_n-l_1]},
\end{equation}
where $l_1\neq l_2$ are free constants,
in the case of F$_{II}$ (\ref{F2q}) via the choice $l_2=l_1\neq 0$ in (\ref{la}) and the substitution
\begin{equation}\label{F2subs}
x_{mn}=\zeta(k_m+h_n-l_1)-\zeta(k_m-l_1)-\zeta(h_n-l_1)-\zeta(l_1),
\end{equation}
and in the case of F$_{IV}$ (\ref{F4q}) via the choice $l_2=l_1=0$ in (\ref{la}) and the substitution
\begin{equation}\label{F4subs}
x_{mn}=\zeta(k_m+h_n)-\zeta(k_m)-\zeta(h_n).
\end{equation}
\end{prop}
\begin{proof}
This can be verified by direct calculation based on Definition \ref{actions}.
To reduce the global compatibility to the local one, it can be confirmed that the linear actions (\ref{la}) give a representation for the group whose relations are encoded in Figure \ref{cdd2}. 

Notice that the substitutions (\ref{F1subs}), (\ref{F2subs}) and (\ref{F4subs}) are of the general form $x_{mn}=F(k_m,h_n)$ for some function symmetric in its two arguments, $F$.
The compatibility for actions $a_2,\ldots,a_i,b_2,\ldots,b_j$ is immediate from the general form of the substitution. 
Compatibility for the trivial actions of $\sigma$,
for instance in the first row of array (\ref{Xdef}), rely on the equality,
\begin{equation}\label{ta}
F(l_1+l_2-h_1-k_m,h_1)=F(k_m,h_1),
\end{equation}
which, because $[-u]=-[u]$ and $\zeta(-u)=-\zeta(u)$, can be verified by inspection.
The trivial actions of $\sigma$ on the first column of (\ref{Xdef}) then follow by symmetry of $F$.
Bear in mind that (\ref{ta}) as written holds only for substitution $F$ in (\ref{F1subs}), the subsequent cases (\ref{F2subs}) and (\ref{F4subs}) require the assumption, as stated in the Proposition, that $l_2=l_1$, and that $l_2=l_1=0$, respectively.

The main calculation, in each of the three cases, is thus reduced to verification of compatibility for the non-trivial action of $\sigma$: the equality
\begin{multline}\label{fq}
F(l_1+l_2-h_1-k_m,l_1+l_2-k_1-h_n) = \\ q(F(k_1,h_1),F(k_m,h_n),F(k_1,h_n),F(k_m,h_1)),
\end{multline}
$m,n>1$, and where $q$ is the corresponding expression (\ref{F1q}), (\ref{F2q}) or (\ref{F4q}).
The calculation to verify that (\ref{fq}) holds identically in $k_1,h_1,k_m,h_n$ is straightforward, but relies on functional identities in the elliptic case.
This is simplified by exploiting the equivalence between (\ref{map}) and systems (\ref{F1a}), (\ref{F2a}) or (\ref{F4a}).
Complement (\ref{fq}) with the equivalent condition obtained by the interchange $h_n\leftrightarrow h_1$,
\begin{multline}\label{fq2}
F(l_1+l_2-h_n-k_m,l_1+l_2-k_1-h_1) = \\ q(F(k_1,h_n),F(k_m,h_1),F(k_1,h_1),F(k_m,h_n)).
\end{multline}
By comparing (\ref{fq}) and (\ref{fq2}) with (\ref{map}), the variables $\{x_1,x_2,x_3\}$ and $\{y_1,y_2,y_3\}$ appearing in (\ref{map}) can be identified with
\begin{equation}\label{allsubs}
\begin{split}
\{F(k_1,h_1),
F(k_m,h_n),
F(l_1+l_2-h_1-k_m,l_1+l_2-k_1-h_n)\},\\
 {\textrm{and }}
\{F(k_1,h_n),
F(k_m,h_1),
F(l_1+l_2-h_n-k_m,l_1+l_2-k_1-h_1)\}.
\end{split}
\end{equation}
The advantage of this reformulation, is that when expressions (\ref{allsubs}) are substituted into (\ref{F1a}), (\ref{F2a}) or (\ref{F4a}), the first equation is satisfied identically, i.e., without any condition on the function $u\mapsto [u]$ or $u\mapsto\zeta(u)$.
In fact the second equation is also satisfied identically, after applying a certain elliptic function identity to each term,  
\begin{equation}
\frac{[u+w][u-w][x+v][x-v]}{[u+x][u-x][w+v][w-v]}=
1-\frac{[u+v][u-v][w+x][w-x]}{[u+x][u-x][w+v][w-v]},
\end{equation}
which is the Riemann relation written in a suitable way,
\begin{equation}
\zeta(u)+\zeta(v)+\zeta(w)-\zeta(u+v+w) = \frac{[u+v][v+w][w+u]}{[u][v][w][u+v+w]},
\end{equation}
and
\begin{equation}
\left(\zeta(u+v)-\zeta(u)-\zeta(v)\right)^2 = \wp(u+v)+\wp(u)+\wp(v),
\end{equation}
respectively.
\end{proof}

\subsection{Connection between the linear representations}
Here, some details are given for calculations that relate the linear actions associated with the two groups.

Combining Propositions \ref{ee} and \ref{decomp} yields linear actions compatible with case $F_I$ of Definition \ref{actions}.
To make this contact and relate it to Proposition \ref{q-solution}, relies on a linear change of variables connecting the substitutions (\ref{esub}) and (\ref{F1subs}).
\begin{prop}[a reformulation of Proposition \ref{ee}]\label{symact}
The association
\begin{equation}
\begin{array}{ll}
l_1=\epsilon_{j+2}-\epsilon_{j+1}-\alpha_0, \ l_2=\epsilon_{j+2}-\epsilon_{j+1}, \\
k_m=\epsilon_{j+2}-\epsilon_{j+m+2}, \quad&m\in\{1,\ldots,i\},\\
h_n=\epsilon_{j+1-n}-\epsilon_{j+1}, \quad&n\in\{1,\ldots,j\},
\end{array}
\end{equation}
defines new variables
\begin{equation}\label{lkh}
l_1,l_2,\ k_1,\ldots,k_i, \ h_1,\ldots,h_j,
\end{equation}
in terms of $\epsilon_0,\ldots,\epsilon_{i+j+2}$.
On the new variables, the linear actions (\ref{w-linear}) are as follows,
\begin{equation}\label{coble-linear}
\begin{array}{rll}
w_0:& l_1\leftrightarrow l_2,& \\
w_{j+1}:& l_2\rightarrow -l_2, \ v\rightarrow v-l_2, & v\in\{l_1,k_1,\ldots,k_i,h_1,\ldots,h_j\},\\
w_j:& h_1\rightarrow -h_1, \ v\rightarrow v-h_1, & v\in\{l_1,l_2,h_2,\ldots,h_j\},\\
w_{j+1-n}:& h_{n-1}\leftrightarrow h_n, & n\in\{2,\ldots,j\},\\
w_{j+2}:& k_1\rightarrow -k_1, \ v\rightarrow v-k_1, & v\in\{l_1,l_2,k_2,\ldots,k_i\},\\
w_{j+1+m}:& k_{m-1}\leftrightarrow k_m, & m \in \{2,\ldots,i\},
\end{array}
\end{equation}
and the substitution (\ref{esub}) becomes (\ref{F1subs}).
\end{prop}
Notice there is one fewer parameter in the list (\ref{lkh}) compared to (\ref{epsilon}).
The composition of actions that gives $\sigma$ in the birational case (\ref{abstr}),
i.e., \[(w_0 w_{j+2} w_{j+1})^2w_jw_{j+1}w_0w_{j+2}w_{j+1}w_j,\] in terms of the linear actions (\ref{coble-linear}), gives the linear action
\begin{equation}\label{alternative}
\begin{array}{ll}
  v\rightarrow -v, & v\in\{l_1,l_2,k_1,h_1\},\\
  k_{m}\rightarrow k_m-l_1-l_2+h_1, & m\in\{2,\ldots,i\},\\
  h_{n}\rightarrow h_n-l_1-l_2+k_1, & n\in\{2,\ldots,j\}.
\end{array}
\end{equation}
In Proposition \ref{q-solution}, the linear action associated with $\sigma$ is is the composition of this one with the map that changes sign of all variables.
It turns out that the linear action (\ref{alternative}) is compatible only in the case $F_I$, but not in the limiting cases $F_{II}$ and $F_{IV}$.

The solution (Proposition \ref{q-solution}) of the group associated with quadrirational maps can be combined with Proposition \ref{norm} to recover, by calculation, linear actions equivalent to the original ones in Proposition \ref{ee}.
\begin{prop}\label{tail}
Consider Propositions \ref{q-solution} and \ref{norm}.
The substitution of (\ref{F1subs}), (\ref{F2subs}) or (\ref{F4subs}) into expression (\ref{zdef}) yields
\begin{equation}\label{zprime}
z_{mn}=\frac{[k'_m+h'_n-l'_1][l'_1][k'_m-l'_2][h'_n-l'_2]}{[k'_m+h'_n-l'_2][l'_2][k'_m-l'_1][h'_n-l'_1]},
\end{equation}
where
\begin{equation}\label{primeshift}
\begin{array}{ll}
l_1'=k_1-k_3,& \\
l_2'=k_2-k_3,& \\
k_m'=k_{m+3}-k_3, & m\in\{1,\ldots,i-3\},\\
h_n'=h_n-l_1-l_2+k_1+k_2, & n\in\{1,\ldots,j\},\\
\end{array}
\end{equation}
where as before the first case (\ref{F1subs}) is written, and the case (\ref{F2subs}) is the same but with $l_2=l_1$, and substitution of (\ref{F4subs}) leads to the above with $l_2=l_1=0$.
The subset $a_2,\ldots,a_i,b_1,\ldots,b_j$ 
of the linear actions (\ref{la})
on variables (\ref{khvars}), 
induce the actions 
\begin{equation}\label{coble-linear2}
\begin{array}{rll}
a_2:& l'_1\leftrightarrow l'_2,& \\
a_3:& l'_2\rightarrow -l'_2, \ v\rightarrow v-l'_2, & v\in\{l'_1,k'_1,\ldots,k'_i,h'_1,\ldots,h'_j\},\\
b_1:& h'_1\rightarrow -h'_1, \ v\rightarrow v-h'_1, & v\in\{l'_1,l'_2,h'_2,\ldots,h'_j\},\\
b_n:& h'_{n-1}\leftrightarrow h'_n, & n\in\{2,\ldots,j\},\\
a_4:& k'_1\rightarrow -k'_1, \ v\rightarrow v-k'_1, & v\in\{l'_1,l'_2,k'_2,\ldots,k'_{i-3}\},\\
a_{m+3}:& k'_{m-1}\leftrightarrow k'_m, & m \in \{2,\ldots,i-3\},
\end{array}
\end{equation}
on the variables $l_1',l_2',k_1',\ldots,k_{i-3}',h_1',\ldots,h_j'$.
\end{prop}
It therefore follows that the linear actions (\ref{coble-linear2}) are compatible with the Coble actions (\ref{induced}) via the substitution (\ref{zprime}), recovering by a different route (\ref{coble-linear}).\\

\noindent
{\bf Acknowledgments.} 
JA acknowledges support from the Australian Research Council, Discovery Grant DP 110104151.
YY is supported by JSPS KAKENHI Grant Number 26287018.

\bibliographystyle{unsrt}
\bibliography{references}

\end{document}